\documentclass[letterpaper, 10 pt,journal]{ieeeconf}   %

\IEEEoverridecommandlockouts                              %

\usepackage{graphics} %
\usepackage{epsfig} %
\usepackage{mathptmx} %
\usepackage{times} %
\usepackage{amsmath} %
\usepackage{amssymb}  %
\usepackage[ansinew]{inputenc} %
\usepackage[T1]{fontenc} %
\usepackage{xcolor}
\usepackage{MnSymbol}
\usepackage{mathtools}
\usepackage{bm}
\usepackage{booktabs}
 \usepackage{microtype} %
\usepackage{units}
\usepackage{siunitx}
\usepackage{tabularx, tabulary}
\usepackage{ragged2e}
\usepackage{multirow}
\usepackage[linesnumbered,ruled,vlined]{algorithm2e}
\let\labelindent\relax %
\usepackage{enumitem}
\usepackage{hyperref}
\usepackage{cite}
\usepackage{standalone}
\usepackage[finalnew]{trackchanges}
\addeditor{SSc}
\newcommand*\Sdelete[1]{\unskip}
\newcommand*\Sadd[1]{#1}

\newcommand{\secret}[1]{\ensuremath{\lsem #1 \rsem}}

\DeclarePairedDelimiter\floor{\lfloor}{\rfloor}
\newcommand{\ie}{i\/.\/e\/.,\/~}%
\newcommand{\eg}{e\/.\/g\/.,\/~}%
\newcommand{\cf}{cf\/.\/~}%
\newcommand{\T}{^{\mathsf{T}}}

\newtheorem{thm}{Theorem}
\newtheorem{cor}{Corollary}
\newtheorem{rem}{Remark}

\let\oldnl\nl%
\newcommand{\nonl}{\renewcommand{\nl}{\let\nl\oldnl}}%

\title{\LARGE \bf
Multi-party computation enables secure polynomial control \newline based solely on secret-sharing
}

\author{Sebastian Schlor, Michael Hertneck, Stefan Wildhagen and Frank Allg{\"o}wer*%
\thanks{*The authors are with the University of Stuttgart, Institute for Systems Theory and Automatic Control, Germany. Funded by Deutsche Forschungsgemeinschaft (DFG, German Research Foundation) under Germany's Excellence Strategy - EXC 2075 - 390740016 and under grant AL 316/13-2 - 285825138. We acknowledge the support by the Stuttgart Center for Simulation Science (SimTech). \{schlor, hertneck, wildhagen, allgower\}@ist.uni-stuttgart.de.}%
}

\begin{document}

\pubid{\begin{minipage}{\textwidth}\ \\[12pt] \copyright 2021 IEEE. Personal use of this material is permitted. Permission from IEEE must be obtained for all other uses, in any current or future media, including reprinting/republishing this material for advertising or promotional purposes, creating new collective works, for resale or redistribution to servers or lists, or reuse of any copyrighted component of this work in other works.\end{minipage}} 

\maketitle
\begin{abstract}

Encrypted control systems allow to evaluate feedback laws on external servers without revealing private information about state and input data, the control law, or the plant.
While there are a number of encrypted control schemes available for linear feedback laws, only few results exist for the evaluation of more general control laws.
Recently, an approach to encrypted polynomial control was presented, relying on two-party secret sharing and an inter-server communication protocol using homomorphic encryption. 
As homomorphic encryptions are much more computationally demanding than secret sharing, they make up for a tremendous amount of the overall computational demand of this scheme.
For this reason, in this paper, we demonstrate that multi-party computation enables secure polynomial control based solely on secret sharing.
We introduce a novel secure three-party control scheme based on three-party computation.
Further, we propose a novel $n$-party control scheme to securely evaluate polynomial feedback laws of arbitrary degree without inter-server communication.
The latter property makes it easier to realize the necessary requirement regarding non-collusion of the servers, with which perfect security can be guaranteed.
Simulations suggest that the presented control schemes are many times less computationally demanding than the two-party scheme mentioned above.

\end{abstract}

\section{INTRODUCTION}

Distributed computing is a rapidly emerging technology both in research and industrial application.
It refers to the concept that multiple computing parties achieve a common goal by interacting via a communication network.
The concept of distributed computing includes also the idea of \emph{Control as a Service}.
Here, as illustrated in Figure~\ref{fig:Network}, control loops are closed via external servers which provide abundant computational power for the evaluation of the control law.
Besides this advantage, there are concerns about data privacy.
Since sensor measurements, the specific controller design, and control actions might be sensitive data, the service providers should not receive any valuable information thereof, as they might not be trustworthy. 
For this reason, secure control schemes have been introduced to guarantee privacy in Control as a Service.
These schemes are typically based on either homomorphic encryption or multi-party computation using secret sharing.
Homomorphic encryption often suffers from high computational costs, while secret sharing requires less computational power but usually more communication.
While previous works on secret sharing-based control consider two non-colluding servers to secretly compute control inputs~\cite{Darup2019, Darup2020}, we introduce in this paper more efficient three-party and $n$-party computation schemes based on secret sharing and show their advantages compared to existing schemes.

\begin{figure}[tb]
	\centering
	\includestandalone[width=0.9\columnwidth, mode=buildnew]{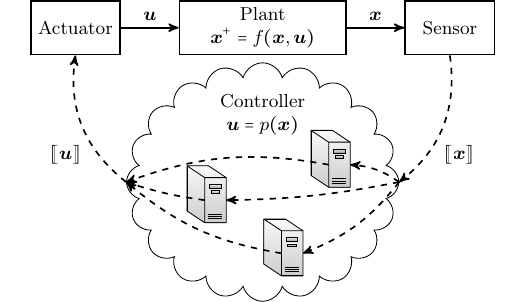}
	\caption{External servers are used in a form of distributed computing to securely control the plant. Thereby, private data is protected by secret sharing.}
	\label{fig:Network}
\end{figure}

Most secure control schemes in the literature are based on homomorphic encryption.
\Sadd{These encryption schemes support a certain subset of mathematical operations to be carried out on encrypted data.}
This means that encryption and decryption preserve the result of these operations on the original data. %
The supported operations can\remove{then} be used to evaluate a control law in secure control. 
\Sadd{Prominent examples for such schemes are the \emph{additively} homomorphic Paillier cryptosystem~\cite{Paillier1999} and the \emph{multiplicatively} homomorphic ElGamal cryptosystem~\cite{Elgamal1985}.}
In principle, \emph{fully} homomorphic encryption schemes such as~\cite{Gentry2010} allow for the computation of any function on finite encrypted data.
However, the computational complexity of these schemes hinders their application to real-world control problems with limited computational power and requirements on the response time.
\Sadd{Therefore, many results are based on the Paillier~\cite{Farokhi2017, Darup2018, Kishida2018, Lin2018, Darup2019d, Darup2019b, Murguia2020} and the ElGamal~\cite{Kogiso2015} cryptosystem.}
Still, with these proposed schemes, multiplications, exponentiations, and modulo operations on large numbers are required to encrypt, decrypt, and process the data.
This results in a large computational demand and long response times.

\pubidadjcol 

A promising approach to overcome this problem is secure multi-party computation using secret sharing (\cf \cite{BenOr1988, Cramer2000}), \Sadd{as}
secret sharing is much less computationally demanding than homomorphic encryption.
Secret sharing protocols split secret data into multiple shares such that individual shares do not reveal anything about the secret. To compute functions on these shared values, the shares are distributed to multiple non-colluding parties which can perform local computations and can communicate with each other.
A two-party computation scheme was used in~\cite{Darup2019} for linear control and was later extended to admit polynomial control laws in~\cite{Darup2020}.
In these two-party computation schemes, additions are easy to evaluate, but products of secret values require computations on secret shares from different parties which are not allowed to be combined.
To be able to compute products in this setup nonetheless,~\cite{Darup2020} proposed to use a joint protocol involving inter-server communication and again homomorphic encryption.
However, this remedy is responsible for 99\% of the computational demand of the overall scheme~\cite{Darup2020}.

Recently, a three-party protocol has been introduced~\cite{Araki2016, Mohassel2018} with which, in contrast to two-party computation, products can be evaluated without having to resort to inter-server communication.
These methods have already been applied to authentication protocols~\cite{Araki2016} and machine learning~\cite{Mohassel2018, Patra2020}, but have not been explored for secure control so far.

In this paper, we investigate secret sharing-based multi-party computation schemes for secure polynomial control.
To this end, we propose two schemes:
\begin{itemize}[leftmargin=5mm]
	\item We introduce a secure polynomial control scheme based on three-party computation.
			For polynomials of up to degree two, no inter-server communication is needed, while higher-order polynomials can be evaluated using additional communications between the servers. In contrast to the two-party scheme~\cite{Darup2020}, symmetric encryptions, which are significantly less computationally demanding than homomorphic encryptions, can be used to this end.
	\item We present a novel $n$-party computation scheme, with which polynomials of arbitrary degree can be evaluated without inter-server communication, and we leverage this scheme for secure polynomial control.
\end{itemize}
With the proposed schemes, state measurements, the control inputs, as well as the controller parameters are kept secret.
Under semi-honest adversaries, the three-party scheme is perfectly or semantically secure, depending on the exact chosen communication structure, and the $n$-party scheme is perfectly secure.
For the $n$-party protocol, the assumptions to guarantee perfect security are less restrictive than for the two- and three-party schemes, since no inter-server communication is needed.
The proposed\remove{three- and $n$-party} control schemes are the first secure polynomial control algorithms that do not \change{require}{need} homomorphic encryption.
As a result, both schemes are less computationally demanding than the scheme in~\cite{Darup2020} by several orders of magnitude, as we demonstrate in a numerical example.

The remainder of this paper is structured as follows.
In Section~\ref{sec:Problem}, the problem under consideration is specified.
After the basics of secret sharing are introduced in Section~\ref{sec:BackPrivacy}, we present how to compute products of $n$ numbers using the three- and $n$-party schemes in Section~\ref{sec:MPC}.
\Sadd{Due to spatial limitations, the proofs of the statements are omitted here.}
In Section~\ref{sec:PolControl}, we demonstrate how the presented schemes are used to realize secure polynomial control.
We illustrate the proposed control schemes with a numerical example in Section~\ref{sec:example}.
Finally, our work is concluded in Section~\ref{sec:conclusion}.

\section{PROBLEM SETUP}\label{sec:Problem}

In this work, we consider a discrete-time system
\begin{align*}
	\bm{x}^{+} = f(\bm{x},\bm{u}),%
\end{align*} %
with the state $\bm{x}\in\mathbb{R}^{n_x}$ and the control input $\bm{u}\in\mathbb{R}^{n_u}$. By $ \bm{x}^+ $, we denote the state at the next time instant.
The system is controlled by a polynomial controller
$\bm{u} =  p(\bm{x})$,
where $p(\cdot)$ is a given real polynomial of degree $d$.
By $x_i$, we denote the $i-$th element of $\bm{x}$ in the following.
The polynomial control law $\bm{u}=p(\bm{x})$ is written as
\Sadd{$\bm{u} = \bm{A} \bm{X}$,}
where $\bm{X}$ is the vector of all (or selected) monomials of $\bm{x}$ up to the finite order $d$, \ie
\begin{align*}
	\bm{X} = \begin{bmatrix}
		1&		x_1&\cdots&x_{n_x}&x_1^2&x_1 x_2&\cdots&x_{n_x}^d
	\end{bmatrix}\T.
\end{align*}
The real matrix $\bm{A}$ is the matrix of coefficients of appropriate size.
Subsequently, for simplicity and without loss of generality for the processing procedure, we consider the scalar input case ${n_u}=1$.
Thus, the control input $u$ is computed as the sum of products with up to $d+1$ factors, considering also the coefficients in $\bm{A}$.
This is what a control device should compute in every time step to provide a control input to the plant.
We consider a configuration where the controller is outsourced to external devices such as cloud computing servers, as depicted in Figure~\ref{fig:Network}.\Sdelete{, sensors, controller and actuators are linked via a communication network.}
The control law should be evaluated on these external servers.
However, we want to hide the control law, the measurements and the resulting control actions, since this data could be confidential.
We consider external servers that honestly execute the protocol but which are curious about the user's data and the control law, which have been referred to as \emph{semi-honest} or \emph{honest but curious}.
Thus, in the sequel of this paper, we present a procedure for the secret evaluation of the polynomial control law $p(\bm{x})$.

\section{SECRET SHARING AND QUANTIZATION}\label{sec:BackPrivacy}

The theory of privacy-preserving computation requires some\remove{fundamental} basics, which are summarized in this section.

\subsection{Secret sharing}\label{sec:SS}
Secret sharing is a method to securely store and process data.
For this purpose, secret data is divided into several shares and the shares are distributed among multiple parties.
In a $(t,n)$ secret sharing scheme, where $2\le t \le n$, the data is split into $n$ shares such that knowledge of at most $t-1$ shares does not reveal anything about the secret, but by combining at least $t$ shares, the secret data can be reconstructed~\cite{Evans2018}.
An overview over several secret sharing variants is given in~\cite{Beimel2011}.
Here, we focus on simple additive secret sharing.%

This work mainly builds on $(2,n)$ secret sharing schemes.
A procedure to create secret shares is given in Algorithm~\ref{algo:ss}.
To encrypt a secret $s$ from a fixed finite number range $\mathbb{N}_{Q} \coloneqq \{0,\dots,Q-1\}$, in this additive secret sharing protocol based on~\cite{Araki2016, Mohassel2018}, $n$ random numbers $s_j \in \mathbb{N}_{Q}$, referred to as components, are drawn uniformly at random with the constraint that $\sum_{j=1}^n s_j \mod Q = s$.
This is achieved by the first and second line of code in Algorithm~\ref{algo:ss}.
Then, the share of Party $p_j$ consists of the tuple $\overline{(s_j)}  \coloneqq (s_\ell|\ell \neq j)$. 
The respective shares are secretly transferred to the $n$ participating parties.
Therefore, each party holds all components of the secret except for the component with the same index as the party.
This yields a $(2,n)$ secret sharing scheme which will be used for further computations.
We see that by combining the shares of at least two parties, the secret can be fully reconstructed, since the combined shares contain all components $s_j, j\in \{1,\dots,n\}$.
However, each individual share does not reveal any information, since its components were generated uniformly at random and are independent of the secret.
Thus, given only one share, each possible secret is equally likely and each share is indistinguishable from another randomly generated tuple.
Therefore, this kind of secret sharing is perfectly secure~\cite{Araki2016}.
By $\secret{s} = (s_1, s_2,\dots,s_n )$ we denote the tuple of components $s_j$ of the secret sharing of $s$.
Whenever we increment indices from a set $ \{1,...,n\} $ in the following, the successor of $n$ is $1$.
Hereafter, when incrementing an index $j$ from a set $ \{1,...,n\} $, we define the successor of $j=n$ to be $j+1= 1$.%

\begin{algorithm}[tb]
	\DontPrintSemicolon %
	\KwIn{Secret $s\in \mathbb{N}_Q$, number of shares $n$}
	\KwOut{Shares $\overline{(s_j)}, j\in \{1,\dots,n\}$ }
		Draw $s_1,s_2 \dots, s_{n-1} \in \mathbb{N}_Q$ uniformly at random\;
		$s_{n} \gets s-\sum_{j=1}^{n-1} s_j \mod Q$\;
		Partition into shares $\overline{(s_j)}  \gets (s_\ell|\ell \neq j)$\;
	\Return{$\overline{(s_j)}, j\in \{1,\dots,n\}$}\;
	\caption{$(2,n)$ secret sharing}
	\label{algo:ss}
\end{algorithm}

\subsection{Quantization and fixed-point number representation}\label{sec:quant}

Most encryption and secret sharing schemes consider numbers from a finite set of integers $\mathbb{N}_{Q}$ as plaintext.
For homomorphic encryption, $Q\in \mathbb{N}$ needs to be \emph{large} to ensure the encryption's security.
For secret sharing, the security is independent of the cardinality of the plaintext space.
The size of $Q$ only depends on the range of numbers to be represented.
\Sadd{For computations on $\mathbb{N}_{Q}$, the following holds:}
\begin{itemize}[leftmargin=5mm]
	\item If $a \{+,-,\cdot\} b \in \mathbb{N}_{Q}$, then  $a \{+,-,\cdot\} b \mod Q = a \{+,-,\cdot\} b$.
	\item If $a \{+,-,\cdot\} b \in \mathbb{Z}$, then $a \{+,-,\cdot\} b \mod Q = (a\mod Q) \{+,-,\cdot\} (b\mod Q) \mod Q$.
\end{itemize}
Therefore, for any polynomial $p$ with integer coefficients, it holds true that $p(a) = p(a\mod Q) \mod Q$ if $p(a), a \in \mathbb{N}_{Q}$.

\Sadd{When dealing with real-valued numbers, they need to be quantized and transformed into an integer representation before secret-sharing.}
First, a measured number $x_r\in \mathbb{R}$ is truncated or rounded to the corresponding value $x$ in a fixed-point representation.
This can be done by a quantizer as
\begin{align}
	x = q_\Delta(x_r) = \begin{cases}
		q_{sat} -\Delta & \text{if } x_r\ge q_{sat} - \Delta\\
		\floor*{\frac{x_r}{\Delta} +\frac{1}{2}} & \text{otherwise}\\
		-q_{sat} & \text{if } x_r \le -q_{sat}
	\end{cases},
\label{eq:quant}
\end{align}
where $q_{sat}\in \mathbb{R}$ is the value of saturation and $\Delta\in \mathbb{R}$ is the precision.
By $\floor*{\cdot}$ we denote the floor function.

If we choose $q_{sat} = \frac{1}{2} \beta^{x_{pre}}$ and $\Delta = \beta^{-x_{post}}$ for a basis $\beta\in \mathbb{N}, \beta>1$, then the resulting number $x$ has at most $x_{pre}\in \mathbb{N}$ digits before the decimal marker and a precision of $x_{post}\in \mathbb{N}$ digits after the decimal marker. %
Now, by multiplying this number $x$ by $\beta^{x_{post}}$, we obtain an integer $\hat{x} \coloneqq x \beta^{x_{post}}$, which can be used for further encryption, sharing or computation.

All such numbers are contained in the set $ \mathbb{N}_{Q}$ with $Q \coloneqq \beta^{x_{pre} + x_{post}}$.
Negative numbers can be represented by their radix complement.
To keep track of the correct scaling factor $\beta^{x_{post}}$, \change{the scaling}{it} must be considered during calculation.
For numbers $a,b,c$ in this fixed-point representation with their respective numbers of fractional digits, the following rules apply:
\begin{itemize}[leftmargin=5mm]
	\item If $a\cdot b = c$, then $a\beta^{a_{post}} \cdot b \beta^{b_{post}} = c \beta^{a_{post} + b_{post}} \mod Q$.
	\item If $a\pm b = c$ and $c_{post} = \max\{a_{post},b_{post}\}$, then $a\beta^{c_{post} - a_{post}} \pm b \beta^{c_{post} - b_{post}} = c \beta^{c_{post}} \mod Q$.
\end{itemize}
With each multiplication, the scaling factor increases, but the modulus $Q$ stays constant.
Therefore, the degree of the polynomial to be evaluated must be taken into account when selecting the number format.
In particular, if the result of the polynomial should have $u_{pre}$ digits before the decimal marker, the indeterminate %
of the polynomial should have $x_{post}$ digits after the decimal marker, and the degree of the polynomial is $d$,
then choosing $Q = \beta^{u_{pre} + (d+1)x_{post}}$, $\Delta = \beta^{-x_{post}}$ and $q_{sat} = \frac{Q}{2}\Delta$ is sufficient.
Hereafter, we implicitly use modular arithmetic whenever handling integers of the set $\mathbb{N}_Q$.
\section{MULTI-PARTY COMPUTATION OF POLYNOMIALS}\label{sec:MPC}
\Sadd{To securely evaluate polynomials, the multi-party computation framework must provide basic primitives for addition and multiplication.
For two secrets $\secret{v_1} = (v_{1,1}, v_{1,2},\dots,v_{1,n})$, $\secret{v_2} = (v_{2,1}, v_{2,2},\dots,v_{2,n})$, and a public constant $c$,
addition of $\secret{v_1}$ and a public constant $c$ can be performed by adding $c$ to one of the random components as $(z_1,\dots,z_n) = (v_{1,1}+c, v_{1,2},\dots,v_{1,n})$.
Addition of two secrets $\secret{v_1}$ and $\secret{v_2}$ is \change{carried out}{done} by adding all components individually as $(z_1,\dots,z_n) = (v_{1,1}+ v_{2,1}, v_{1,2}+ v_{2,2},\dots,v_{1,n}+ v_{2,n})$.
Multiplication of $\secret{v_1}$ and a public constant $c$ is easily performed by multiplying each component by the constant as $(z_1,\dots,z_n) = (v_{1,1} c, v_{1,2} c,\dots,v_{1,n} c)$.
\add{As reported in}~\cite{Darup2019}, with these functionalities, affine functions such as matrix multiplications and affine control laws can be evaluated if one of the factors is publicly known.}
To securely evaluate higher-order products while ensuring the secrecy of all factors, multiplication of multiple secrets is also necessary.
While the previous primitives are similar for all secret sharing schemes of this additive type, they differ when multiplying secret numbers.
As mentioned in Section~\ref{sec:RelWork}, existing two-party multiplication protocols as in~\cite{Darup2020} require the secret exchange of shares between parties using homomorphic encryption.
In this section, first, we present how to multiply $n$ factors securely using three-party computation, and secondly, we propose a novel $n$-party scheme for this purpose.

\subsection{Three-party computation}\label{sec:three}
Secure three-party computation can be carried out as in the framework of~\cite{Araki2016, Mohassel2018}.
There, three non-colluding parties and a $(2,3)$ secret sharing scheme are used to multiply secret numbers.
This enables to multiply two factors locally without communicating shares between the servers.
Multiplication of $\secret{v_1} = (v_{1,1}, v_{1,2},v_{1,3})$ and $\secret{v_2} = (v_{2,1}, v_{2,2},v_{2,3})$\remove{ in a $(2,3)$ sharing} reads
\begin{multline*}
		\secret{v_1}\secret{v_2} = \underbrace{(v_{1,1}v_{2,1} + v_{1,1}v_{2,2} + v_{1,2}v_{2,1})}_{\eqqcolon z_3}\\
	    +  \underbrace{(v_{1,2}v_{2,2} + v_{1,2}v_{2,3} + v_{1,3}v_{2,2})}_{\eqqcolon z_1}
		+ \underbrace{(v_{1,3}v_{2,3} + v_{1,3}v_{2,1} + v_{1,1}v_{2,3})}_{\eqqcolon z_2}.
\end{multline*}
As described in Subsection~\ref{sec:SS}, in the $(2,3)$ sharing, each party obtains two of the three \Sdelete{random} components of each secret.
As a consequence, Party $p_3$ has access to the shares $(v_{1,1}, v_{1,2})$ and $(v_{2,1}, v_{2,2})$, which are necessary for calculating $z_3$.
Similarly, parties $p_1$ and $p_2$ can locally compute $z_1$ and $z_2$, respectively.
As a result, each party holds one component of the product $(z_1, z_2, z_3) \eqqcolon \secret{z} = \secret{v_1} \secret{v_2}$ in a $(3,3)$ sharing.
If this result should serve as a factor for another multiplication, again a $(2,3)$ sharing of it has to be created.
To do this, a zero-sharing among the servers, \ie random numbers $a_1,a_2,a_3\in \mathbb{N}_Q$ which fulfill  $a_1 +a_2 +a_3=0$, is created. As described in~\cite[Section 2.2]{Araki2016}, this can be done semantically secure without communication, or perfectly secure with communication between the parties.
Each party holds only one of these numbers.
Then, the servers can circularly transmit their obfuscated component $\tilde{z}_j = z_j +a_j$ to the subsequent party. 
This secure transmission can be efficiently realized with a symmetric encryption. 
Thus, each party again holds two of the three components of $z$ as its share, \ie $(\tilde{z}_j,\tilde{z}_{j-1})$, resulting in a $(2,3)$ sharing.
The obfuscation of the shares with the random zero-sharing ensures that the shares of $z$ are again uncorrelated to the shares of $v_1$ and $v_2$, which is necessary for perfect security~\cite{Evans2018}.
Based on this concept, products of $n$ factors can be computed as in Algorithm~\ref{algo:threeProduct}.
From \cite{Araki2016}, we inherit the following result on security.
\begin{thm}[{\cf \cite[Theorem 3.7]{Araki2016}}]
	Under secure communication channels and non-colluding servers, the three-party multiplication protocol is perfectly secure if the zero-sharing is perfectly secure, and semantically secure if the zero-sharing is semantically secure.
\end{thm}

\begin{algorithm}[tb]
	\DontPrintSemicolon %
	\KwIn{Factors $v_1, v_2,\dots, v_{n}$}
	\KwOut{Product $z = v_1 v_2 \cdots v_{n}$}
	\hrule
	\nonl \textbf{At the distributor:}\;
	\For{$k \gets 1$ {\upshape to} $n$} {
		\parbox[t]{\dimexpr\columnwidth-\leftmargin-\labelsep-\labelwidth-15.5pt}{Use Algorithm~\ref{algo:ss} to create $(2,3)$ secret sharing $\secret{v_k} = (v_{k,1}, v_{k,2}, v_{k,3})$ with shares $\overline{(v_{k,j})}, j\in \{1,2,3\}$}\;
		Send each share $\overline{(v_{k,j})}$ to server $p_j$\;
	}
	\hrule
	\nonl\textbf{At each server $p_j$:}\;
	$(y_1,y_2) \gets \overline{(v_{1,j})}$\;
	\For{$k \gets 1$ {\upshape to} $n-2$} {
		$y_1\gets y_1 v_{k+1,j+1} + y_1 v_{k+1,j-1} + y_2 v_{k+1,j+1} $\;
		$a_j \gets$ component of new zero-sharing\;
		$y_1\gets y_1 + a_j$\;
		Send $y_1$ to server $p_{j+1}$\;
		$y_2\gets $ received $y_1$ from server $p_{j-1}$\;
	}
	$z_j\gets y_1 v_{n,j+1} + y_1 v_{n,j-1} + y_2 v_{n,j+1} $\;
	Send resulting share $z_j$ to collector\;
	\hrule
	\nonl\textbf{At the collector:}\;
	$z\gets \sum_{j=1}^{3} z_j$\;
	\Return{$z$}\;
	\caption{Secure $n$-factor product by three-party computation}
	\label{algo:threeProduct}
\end{algorithm}%

\subsection{$n$-party computation}\label{sec:nParty}

In this subsection, we present a novel approach to compute a product $\secret{z} = \secret{v_1}\secret{v_2}\cdots \secret{v_{n}}$ of $n$ factors simultaneously \Sdelete{by generalizing the existing three-party computation to $n$-party computation}.
Inspired by the fact that three parties can directly compute a product of two factors \Sdelete{(\cf Subsection~\ref{sec:three})}, we exploit that a product of $ n $ factors can be computed directly by $ n+1 $ parties.
After creating $(2,n+1)$ secret sharings of the factors at a distributor, \eg the sensor, the product of the shared factors can be written as
\begin{align}
		z ={} &{} \phantom{\cdot} (v_{1,1} + \dots + v_{1,n+1}) \cdots  (v_{n,1} + \dots + v_{n,n+1})\notag\\
	={} & \sum_{j_1=1}^{n+1} \sum_{j_2=1}^{n+1} \cdots \sum_{j_n=1}^{n+1} v_{1,j_1} v_{2,j_2} \cdots v_{n,j_n}  \eqqcolon \sum_{j=1}^{n+1} z_j. \label{eq:SumSum}
\end{align}
As there are $n+1$ parties which each hold $n$ components of each factor and there are always $n$ factors in a summand $ v_{1,j_1}\cdots v_{n,j_n}$ of~\eqref{eq:SumSum}, for every summand there is at least one party capable of computing it. 
The indices of summands that one party $ p_j $ should compute is captured by the subset $\mathcal{I}_j \subset \{1,..., n+1\}^n$, which is assigned before the online phase.
Typically, it is advantageous if the computations are evenly distributed among the parties.
\begin{rem}\label{rem:distribution}
One possible way to evenly distribute the summands to be computed is the following.
Fix the index $ j_{n} = 1 $. For every index combination $ (j_1,\dots, j_{n-1})\in \{1,\dots,n+1\}^{n-1}$, assign the summand $ v_{1,j_1}\cdots v_{n-1,j_{n-1}} v_{n,j_n}$ to one of the parties $ p_j $ with $j \in \{1,\dots,n+1\}\backslash \{j_1,\dots, j_n\}$. We add the indices of this summand to the index set $\mathcal{I}_j$.
Then, recursively for all $ \ell\in \{1,\dots,n\}$, the summand with indices $ \{j_1+\ell,\dots, j_n+\ell\}$ is assigned to party $p_{j+\ell}$.
\end{rem}
Each party $p_j$ sums its computed summands to its share $z_j$ of the result.
\Sdelete{By combining the shares of all parties, the result can be restored.}
Thus, an $(n+1,n+1)$ sharing of $z$ is created.
The shares are sent to a collector, \eg the actuator, where the sum of $z_j$ is calculated to restore the result $z$.
Our proposed $ n $-party multiplication scheme is formalized in Algorithm~\ref{algo:multiProduct}.
In contrast to sequential two- or three-party computation to compute higher-order products, by using this $ n $-party computation scheme, no communication except for distributing and collecting the shares is necessary.
This brings the benefit that a) communication delays due to inter-server communication are avoided and b) that computations can be carried out in parallel.
Since the individual parties do not exchange their shares, for this scheme, we can guarantee perfect security as defined in~\cite[Definition~3.1]{Araki2016}.
\begin{thm}
	Under secure communication channels and non-colluding servers, the $n+1$-party multiplication protocol is perfectly secure.%
\end{thm}%
\begin{rem}
	The common assumption of non-colluding servers in secure multi-party computation can be rather restrictive.
	Its realization is especially difficult if the parties have to communicate while executing the protocol, since then, the user is forced to inform each party about the other parties. 
	However, in this $n$-party protocol, the individual servers do not need to interact with each other, and hence, there is no need to inform them about the other parties.
	Thus, even if servers wanted to collude, they would not be able to because they would not know which other parties to collude with.
	
\end{rem}

\begin{algorithm}[tb]
	\DontPrintSemicolon %
	\KwIn{Factors $v_1, v_2,\dots, v_{n}$}
	\KwOut{Product $z = v_1 v_2 \cdots v_{n}$}
	\hrule
	\nonl\textbf{At the distributor:}\;
	\For{$k \gets 1$ {\upshape to} $n$} {
		\parbox[t]{\dimexpr\columnwidth-\leftmargin-\labelsep-\labelwidth-15.5pt}{Use Algorithm~\ref{algo:ss} to create $(2,n+1)$ secret sharing $\secret{v_k} = (v_{k,1}, v_{k,2}, \dots, v_{k,n+1})$ with shares $\overline{(v_{k,j})}, j\in \{1,\dots,n+1\}$}\;
		Send each share $\overline{(v_{k,j})}$ to server $p_j$\;
	}
	\hrule
	\nonl\textbf{At each server $p_j$:}\;
	$z_j \gets 0$\;
	\For{{\upshape each} $(j_1,\dots, j_{n})$ {\upshape in} $\mathcal{I}_j$} {
		$z_j \gets z_j +  v_{j,j_1}\cdots v_{j,j_n}$\;
	}
	Send resulting share $z_j$ to collector\;
	\hrule
	\nonl\textbf{At the collector:}\;
	$z\gets \sum_{j=1}^{n+1} z_j$\;
	\Return{$z$}\;
	\caption{Secure $n$-factor product by $n+1$-party computation}
	\label{algo:multiProduct}
\end{algorithm}

\section{SECURE POLYNOMIAL CONTROL USING MULTI-PARTY COMPUTATION}\label{sec:PolControl}

Now that we have provided all the necessary primitives for the evaluation of polynomials, we show how polynomial control laws can be evaluated securely.
We present the procedure for three-party computation in parallel with the procedure for $n$-party computation, since they are very similar in most parts.
An overview over the participating units and the structure of the $n$-party computation is given in Figure~\ref{fig:Overview}.

As described in Section~\ref{sec:Problem}, we want to evaluate a polynomial control law $u_\mathrm{r}=p_\mathrm{r}(\bm{x}_\mathrm{r})$, which can be written as
$u_\mathrm{r} = \bm{A}_\mathrm{r} \bm{X}_\mathrm{r}$,
where $\bm{X}_\mathrm{r} = \begin{bmatrix}
		1&		x_{\mathrm{r}1}&\dots&x_{\mathrm{r}{n_x}}&x_{\mathrm{r}1}^2&x_{\mathrm{r}1} x_{\mathrm{r}2}&\dots&x_{\mathrm{r}{n_x}}^d
	\end{bmatrix}\T $
is the vector of monomials of the real vector $\bm{x}_\mathrm{r}$ up to order $d$, and $\bm{A}_\mathrm{r}$ is the vector of coefficients.
For illustration purposes, we demonstrate the necessary steps on the scalar polynomial control law $u_\mathrm{r} = a_{\mathrm{r}1}x_{\mathrm{r}1}^2 + a_{\mathrm{r}2}x_{\mathrm{r}1} + a_{\mathrm{r}3}x_{\mathrm{r}2} + a_{\mathrm{r}4}$.

First, we describe the preparation in the offline phase.
Since secret sharing operates on integers, we have to create an integer representation of all numbers, and hence, a quantization scheme as in~\eqref{eq:quant} is determined.
Then, the polynomial coefficients $ \bm{A}_\mathrm{r} $ are already quantized offline to obtain $ \bm{A} = q_\Delta(\bm{A}_\mathrm{r}) $.
In each time step, this quantizer will also be used to quantize the state as $ \bm{x} = q_\Delta(\bm{x}_\mathrm{r}) $.
Thus, the quantized polynomial $u = a_{1}x_{1}^2 + a_{2}x_{1} + a_{3}x_{2} + a_{4}$ will be evaluated.
Further, a distribution scheme is designed, which determines the assignment of secret shares and computations to the servers.
\Sdelete{For that purpose, each summand of the polynomial is treated separately.} This leads to share distributions as in Algorithm~\ref{algo:threeProduct} or~\ref{algo:multiProduct} for each summand in the control law.
In our example, the constant $ a_{4} $ is sent directly to the actuator. For two-factor products, such as the summands $ a_{2}x_{1}$ and $ a_{3}x_{2} $, in both the three- and $n$-party scheme, a $(2,3)$ sharing of each factor and three servers are needed.
For products of more than two factors, as $ a_{1}x_{1}^2 $, in the $n$-party scheme, a $(2,4)$ sharing of each factor is created and distributed to four servers as indicated in Figure~\ref{fig:Overview}.
In the three-party algorithm, again $(2,3)$ sharings of each factor are distributed to three servers.
\begin{rem}\label{rem:lessServers}
	The same server can be involved in \change{the computation of}{computing} different summands of the polynomial without gaining access to \change{sensible}{sensitive} information.
		Therefore, the actual number of servers needed in the $n$-party scheme is $d+2$ if $d$ is the degree of the polynomial, \change{since}{as} there are at most $n=d+1$ factors in one summand.
		In our example of degree $d=2$, four servers are needed, since the products $ a_2x_1 $ and $ a_3x_2 $ can be computed on the same servers as the highest-order product $a_1x_1^2$. 
		For the three-party scheme, always three servers are sufficient.
\end{rem}

After the quantization and distribution schemes are prepared, the online phase starts.
In every time step, the sensor measures the state $ \bm{x}_\mathrm{r} $ of the system, which is then quantized to $\bm{x}$ in a fixed-point or integer representation.
Then, for each summand of the polynomial, Algorithm~\ref{algo:threeProduct} or~\ref{algo:multiProduct} is executed. %
The servers send their shares $ z_j $ of the result to the collector via secure communication channels.
The collector receives the shares and computes the sum to reveal the secret result $u$.
Then, the quantized control input $u$ is applied to the plant.

Thus, \Sdelete{to summarize,} the polynomial control law is evaluated in its quantized form, and all multiplications are performed on the servers.
The sensor device needs to create random numbers from a finite set and compute modular sums and differences such that the \Sdelete{random} components add up to the secret.
To reveal the result, the actuator device only \Sdelete{needs to add} \Sadd{adds} the received shares.
Further, we can state the following corollary for the security of the polynomial control scheme, inherited from the security of the three- or $n$-party multiplication scheme.
\begin{cor}
	The polynomial control scheme has the same security properties as the multiplication scheme used therein.
\end{cor}

\begin{figure}[tb]
	\centering
	\includestandalone[width=0.95\columnwidth, mode=buildnew]{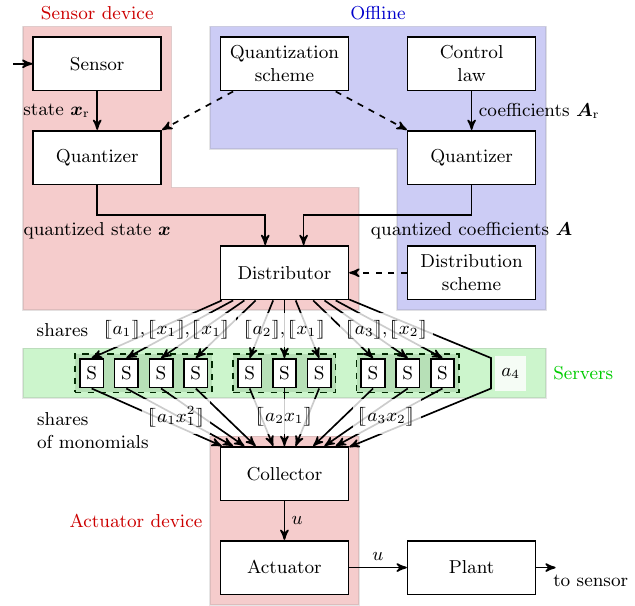}
	\caption{
		Structure of the evaluation of polynomial control laws using secure $n$-party computation.
		The distribution of computations to different servers is illustrated using the example of $u = a_{1}x_{1}^2 + a_{2}x_{1} + a_{3}x_{2} + a_{4}$.}
	\label{fig:Overview}
\end{figure}

\section{NUMERICAL EXAMPLE}\label{sec:example}

To illustrate the effectiveness of our approach, we simulate the multi-party computation protocols from Section~\ref{sec:PolControl} in a numerical example.
To be able to compare our schemes to state-of-the-art secure polynomial control protocols, we adopt the example from~\cite{Darup2020}.
We thus consider the\remove{ polynomial} system
\begin{align*}
	\dot{\bm{x}}=\frac{1}{T}\begin{bmatrix}
		-x_{1}+x_{1} x_{2}+x_{2} u \\
		x_{1}+2 x_{2}+x_{1}^{2}+x_{1}^{2}x_{2}+u
	\end{bmatrix}
\end{align*} %
with time scaling factor $T=1000$.
As described in~\cite{Darup2020}, this system is a scaled version of~\cite[Example 3]{Valmorbida2013} and a stabilizing control law derived therein is
\begin{align*}
	g(\bm{x})=&{} 1.6973 x_{1}-12.2838 x_{2}-0.2122 x_{1}^{2}-2.6975 x_{1} x_{2}\notag\\
	&{} +1.9631 x_{2}^{2} +0.7721 x_{1}^{3}-4.6034 x_{1}^{2} x_{2}\\
	&{} +0.2959 x_{1} x_{2}^{2}-2.3850 x_{2}^{3}.\notag
\end{align*}
This control law is to be evaluated by our multi-party computation frameworks and the input will be applied to the system in a zero-order hold fashion.
Therefore, first a quantization scheme is chosen.
For an easier interpretation, we choose $\beta=10$.
To make the complexity comparable to~\cite{Darup2020}, we \Sdelete{also} choose a precision of $x_{post} = 2$ fractional digits.
Further, as discussed in~\cite{Darup2020,Valmorbida2013}, the domain of attraction of the closed-loop system is constrained by $|x_1|<6, |x_2|<6$.
In this region, the control input is bounded by $|u|<2000$.
Thus, four digits before the decimal marker are sufficient in the result and we choose $u_{pre}= 4$.
We observe that the degree of the control law is $d = 3$.
In accordance with Section~\ref{sec:quant}, we choose $Q = 10^{4 + (3+1)2} = 10^{12}$, $\Delta = 10^{-2}$ and $q_{sat} = \frac{Q}{2}\Delta = \num{5e9}$.

To estimate the average time our algorithms require to compute the control input based on the quantized state, we decomposed them into elementary operations.
Then, we used the average computation times of these operations reported in~\cite[Table II]{Darup2019}, which were also used in~\cite{Darup2020}, to obtain the average computation times of our algorithms for this example.
The results are given in Table~\ref{tab:results}.
We observe that our three- and $n$-party protocols are about three orders of magnitude faster than the two-party scheme in~\cite{Darup2020}, where an average time of \SI{762.13e-3}{\second} was reported.
In particular, our schemes are about 6000 and 1000 times faster, which is mainly due to the fact that they do not need to employ homomorphic encryption for secure multiplication.
\begin{table}[tb] %
	\caption{Average computation times of the polynomial control schemes.}
	\label{tab:results}
	\centering
	\begin{tabular}{ l r r r r}%
		\toprule
		Scheme &  Sensor & Server& Actuator & Total \\
		\midrule
		Three-party [\SI{e-6}{\second}]  &  43.4 & 54.2 & 27.9 &  125.5\\
		 $n$-party [\SI{e-6}{\second}] &  67.6 & 622.9 & 40.7& 731.2\\
		\bottomrule
	\end{tabular}
\end{table}

\section{CONCLUSION AND OUTLOOK}\label{sec:conclusion}

\Sadd{In this work, we presented two multi-party computation protocols for secure evaluation of polynomial control laws on external servers. 
Both schemes can be evaluated much faster compared to previous results, as the presented methods do not rely on homomorphic encryption, but solely on secret sharing.
Our novel $n$-party computation protocol can evaluate polynomial control laws in one round without inter-server communication. 
Due to this property, the common assumption of non-colluding servers can be ensured more easily, since servers do not need to be informed about other parties.
We \change{proved that the two protocols are secure, and demonstrated their practicality}{demonstrated the practicality of the two protocols} in a numerical example.}

Further research includes extensions of the presented protocols.
\Sadd{Conversion between algebraic and binary shares as in \cite{Mohassel2018} and possible bit extraction enables comparisons of numbers, such that piecewise polynomial functions could be evaluated.}
This extends the scope of possible applications to more general control strategies such as model predictive control.
Regarding security, the three-party framework can also be extended to the case with one malicious cloud~\cite{Mohassel2018}.

\bibliographystyle{IEEEtran}
\bibliography{UsedInCDC2021}

\end{document}